\documentclass{lmcs}

\usepackage{enumerate}
\usepackage{hyperref}
\usepackage{amssymb}
\usepackage{graphicx,epsfig,epstopdf}
\usepackage[usenames,dvipsnames,svgnames]{xcolor}

\theoremstyle{plain}\newtheorem{remark}[thm]{Remark}
\theoremstyle{plain}
\theoremstyle{plain}

\newcommand{\op}{\operatorname{op}}
\newcommand{\da}{\mathord{\downarrow}}
\newcommand{\ua}{\mathord{\uparrow}}

\newcommand{\cl}{\operatorname{cl}}

\newcommand{\bigsup}{\bigvee}
\newcommand{\biginf}{\bigwedge}

\newcommand\twoheaddownarrow{\mathord{\rotatebox[origin=c]{90}{$\twoheadleftarrow$}}}
\newcommand{\dda}{\twoheaddownarrow}

\newcommand{\finite}[1]{\operatorname{fin}(#1)}

\begin{document}

\title[(Join-continuity + Hypercontinuity = Prime-continuity]{Join-continuity + Hypercontinuity = Prime Continuity}

\author[W. K. Ho]{Weng Kin Ho}	
\address{National Institute of Education, Nanyang Technological University, 1 Nanyang Walk, Singapore 637616}	
\email{wengkin.ho@nie.edu.sg}  

\author[A. Jung]{Achim Jung}	
\address{School of Computer Science,
         The University of Birmingham, Edgbaston, Birmingham, B15 2TT, United Kingdom}	
\email{axj@cs.bham.ac.uk}  

\author[D. Zhao]{Dongsheng Zhao}	
\address{National Institute of Education, Nanyang Technological University, 1 Nanyang Walk, Singapore 637616}	
\email{dongsheng.zhao@nie.edu.sg}  





\begin{abstract}
  \noindent A remarkable result due to Kou, Liu \& Luo states that the condition
  of continuity for a dcpo can be split into quasi-continuity and
  meet-continuity. Their argument contained a gap, however, which is
  probably why the authors of the monograph \emph{Continuous Lattices and
    Domains} used a different (and fairly sophisticated) sequence of
  lemmas in order to establish the result. In this note we show that
  by considering the Stone dual, that is, the lattice of Scott-open
  subsets, a straightforward proof may be given. We do this by showing
  that a complete lattice is prime-continuous if and only if it is
  join-continuous and hypercontinuous. A pleasant side effect of this
  approach is that the characterisation of continuity by Kou, Liu \&
  Luo also holds for posets, not just dcpos.
\end{abstract}

\maketitle

\section{Introduction}
\label{sec: intro}
The notion of continuity can be said to be the very foundation of the
whole of domain theory.  The pioneering class of domains, continuous
lattices, introduced by Dana Scott in~\cite{scott72} was intended for
applications in theoretical computer science~\cite{scott73}.  In these
applications, the phenomenon of approximation can be formalized in any
partial order by using the way-below relation $\ll$.  A non-empty
subset $D$ of a poset $P$ is \emph{directed} if two elements in $D$
always have an upper bound in $D$.  For any $x, y \in P$, $x \ll y$ if
for any directed set $D$, $\bigsup D \geq y$ implies $D~ \cap \ua x \neq
\emptyset$ whenever $\bigsup D$ exists.  Roughly speaking, one may
view $x \ll y$ as `$x$ is an approximation of $y$'.  We say that a
poset $P$ is \emph{continuous} if for each $x \in P$, there are enough
elements approximating it in the sense that $\dda x := \{ p \in P \mid
p \ll x\}$ is directed and $\bigsup \dda x = x$.  Researchers in
continuous lattices soon extended their study to more general classes
of partial orders, ranging from directed complete posets
(\emph{dcpos}, for short) (\cite{gierzetal03}) to just posets (see,
for example, \cite{wangetal10}), and hence the birth of the term
\emph{domain} which is meant to include all partially ordered
structures that are equipped with some form of approximation.

In recent years the development of domain theory saw the evolution of
various kinds of continuous structures. One thread of generalisation
involves the replacement of directed subsets by other kinds of subsets
(called \emph{$Z$-sets}) so that most of the existing results in
domain theory carry over to a more general setting, initiated by the
work of~\cite{wrightetal78} and followed by later works such as
\cite{bandelterne83, baranga96, menon96, erne99, ernezhao01}. In fact,
Raney's characterization~(\cite{raney52}) of completely distributive
lattices as supercontinuous complete lattices (also called \emph{prime
  continuity} in~\cite[p.107]{abramskyjung94}; see also \cite[Exercise
8.3.15]{goubaultlarrecq16}) is a forerunner of this generalisation to
$Z$-subset systems. Connections have been made with mainstream domain
theory; for instance when directed sets are replaced by Scott-closed
sets, an order-theoretic characterization of the Hoare powerdomain was
obtained using the notion of $C$-continuity~(\cite{hozhao09}).

Another distinctive thread of generalisation began with the invention
of quasicontinuous domains by Gierz, Lawson, and Stralka
in~\cite{gierzetal83}.  Instead of changing the directed subsets, the
idea was to extend the way-below relation between two points in a
poset to that between two finite subsets.  More precisely, for any two
nonempty subsets $F$ and~$G$ of a poset~$P$, define $F \ll G$ if
whenever an existing supremum of a directed set $D$ is in $\ua F$,
then $D \cap {\ua G} \neq \emptyset$. A poset $P$ is said to be
\emph{quasicontinuous} if for all $x \in P$,  $\finite{x} := \{\ua F
\mid F \text{ is a finite subset of } P,~F \ll \{x\}\}$ is a directed
family of subsets of the poset with respect to reverse inclusion, and
$\bigcap \finite{x} = \ua x$.  Unlike the $Z$-generalisation, this
current trend in domain theory to develop a more complete
understanding of quasicontinuity
has had a powerful impact on the development of domain theory itself. We
highlight three important instances of this: (1)~The Scott topology of
quasicontinuous domains are exactly the hypercontinuous lattices;
also, the theory of quasicontinuous domains makes connection with the
Scott and Lawson topologies, \cite{gierzlawson81,gierzetal83}.
(2)~Besides hypercontinuity, meet-continuity for dcpos is yet another
example of a relatively novel variant of continuity. Invented
initially as a generalisation from complete lattices to dcpos, this
new notion turns out to have close connections with Hausdorff
separation, quasicontinuity, continuity and Scott-filter bases,
\cite{kouliuluo03}. (3)~A recent and significant milestone is the
introduction of $\mathbf{QRB}$ domains (\emph{quasi-retracts of
  bifinite domains}) by Jean Goubault-Larrecq in his attempt to make
progress with the Jung-Tix problem, \cite{goubaultlarrecq12}. Of
course, one expects fusion of the $Z$ and `quasi' approaches as
already witnessed by \cite{xuliu03,xuetal05}.

With the prolific emergence of new kinds of continuity in domain
theory, it is important to understand relationships among them.
Here are some examples of known connections:
\begin{enumerate}
\item continuity + $C$-continuity = prime continuity (= complete
  distributivity) holds for complete lattices
(\cite[Theorem 3.11]{hozhao09})
\item meet continuity + quasicontinuity = continuity
holds for dcpos
\item prime continuity $\implies$ continuity $\implies$ meet continuity
\item prime continuity $\implies$ hypercontinuity $\implies$ continuity
\end{enumerate}

This paper is specifically about Equation~(2), first stated
in~\cite[p. 122, Theorem 2.5]{kouliuluo03}.\footnote{Unfortunately,
  the proof given in \cite{kouliuluo03} contains a faulty
  argument as we shall explain in Section~\ref{sec: main results}.}
We ask the following two questions:
\begin{enumerate}
\item Is it possible to provide a proof of~(2) that exploits our
  knowledge of the structure of the lattice of Scott-open subsets?
\item Is the statement still true when we extend it to the class of all posets?
\end{enumerate}

Indeed, below we establish the following for
complete lattices:
\begin{equation}
\text{join continuity} + \text{ hypercontinuity } = \text{prime continuity}
\end{equation}
Based on this result, we will be able to answer affirmatively our two
questions.

\section{Preliminaries}
\label{sec: prelim}
We gather here all the definitions and results that we need in
Section~\ref{sec: main results}, leaving out all proofs.  However, we
take extra care to ensure that none of these make use, directly or
indirectly, of Lemma III-2.10 and Proposition III-2.10
in~\cite{gierzetal03} so as to make this note as self-contained as
possible.

A subset $U$ is \emph{upper} if $U = \ua U$, where $\ua U : =\{p \in P
\mid \exists u \in U.~u \leq p\}$. We call a subset $U$ of a poset $P$
\emph{Scott open} if (i) $U$ is upper and (ii) whenever an existing
supremum of a directed set $D$ is in $U$, then already $D \cap U \neq
\emptyset$. The collection of Scott opens of $P$, denoted by
$\sigma(P)$, defines a topology on it, termed as the \emph{Scott
  topology}. A subset $C$ of $P$ is \emph{Scott-closed} if $P
\backslash C \in \sigma(P)$. We use $\Gamma(P)$ (or $\sigma^{\op}(P)$)
to denote the collection of Scott-closed sets of $P$. Both $\sigma(P)$
and $\Gamma(P)$, when ordered by set inclusion, become complete
lattices, and we overload the symbols $\sigma(P)$ and $\Gamma(P)$ to
refer to these complete lattices.

A completely distributive lattice $L$ is a complete lattice in which
the following, so-called complete distributive law, is satisfied: for
all families $(u^i_j)_{j \in J_i}$, one for each $i \in I$,
$\biginf_{i \in I} \bigsup_{j \in J_i} u^i_j = \bigsup_{f \in \Pi_{i
    \in I} J_i} \biginf_{i \in I} u^i_{f(i)}$.
We have the following well-known result:
\begin{thm} (\cite{hoffmann81}) \label{thm: spectra of completely distributive}
The following statements are equivalent for a poset $P$:
\begin{enumerate}
\item $P$ is a continuous poset.
\item $\sigma(P)$ is a completely distributive lattice.
\end{enumerate}
\end{thm}
Define on a complete lattice $L$ the \emph{way-way-below} relation
$\triangleleft$ as follows: $u \triangleleft v$ if for any $S
\subseteq L$, whenever $\bigsup S \geq v$ then $u \in {\da S}$. A
complete lattice $L$ is said to be \emph{prime continuous} if every
element in $L$ is the least upper bound of all elements way-way-below
it. It is straightforward to show that a complete lattice $L$ is prime
continuous if and only if it is completely distributive, and this was
first established in~\cite{raney52}, albeit using a very different
formulation. In view of the focus of this paper, we use the term
`prime continuity' in preference to `complete distributivity'.

Besides prime continuity, meet continuity is the property that one
encounters very often in domain theory. A complete lattice~$L$ is
\emph{meet continuous} if for all $x \in P$ and all directed
subsets~$D$ of~$P$, it holds that $x \wedge \bigsup D = \bigsup \{x
\wedge d \mid d \in D\}$. This property can be characterised by the
Scott topology as $x \in \cl_\sigma(\da x ~\cap \da D)$ whenever $x
\leq \bigsup D$. Since the meet operator is not involved, this
topological property of meet continuity can be used to give a natural
extended meaning to \emph{meet continuity} in the more general setting
of dcpos. Dcpos which enjoy meet continuity are called \emph{meet
  continuous dcpos.}\footnote{This definition was first proposed
 in~\cite{kou98}.} This definition quickly generalises to \emph{meet
  continuous posets} where the phrase `all directed subsets of $P$' is
replaced by `all directed subsets of $P$ whose suprema exist'
(\cite{maoxu09}).

It is very natural to ask if the meet continuity of a poset $P$ can be recognized from the properties of $\sigma(P)$.  The answer is yes:
\begin{thm} (\cite{maoxu09}, Theorem 3.8) \label{thm: scott open of meet cont is join cont}
The following statements are equivalent for a poset~$P$:
\begin{enumerate}
\item $P$ is meet continuous.
\item $\sigma(P)$ is join continuous.
\item $\sigma^{\op}(P)$ is a frame.
\end{enumerate}
\end{thm}
Here, a complete lattice $L$ is said to be \emph{join continuous} if
for all $x \in L$ and all $S \subseteq L$, we have $x \vee \biginf S =
\biginf \{x \vee s \mid s \in S\}$. A \emph{frame} is just the order
dual of a join continuous complete lattice. Since prime continuity is
equivalent to complete distributivity, it is immediate that
\begin{equation} \label{eq: prime implies join}
\text{prime continuity} \implies \text{join continuity}
\end{equation}
and
\begin{equation}
\text{prime continuity} \implies \text{frame}.
\end{equation}

A third type of continuity that is central to our present discussion is hypercontinuity.  Analogous to continuity, this concept is defined via a certain auxiliary relation on a complete lattice $L$: $x \prec y$ if whenever the intersection of a nonempty collection of upper sets is contained in $\ua y$, then the intersection of finitely many is contained in $\ua x$.
A complete lattice $L$ is called \emph{hypercontinuous} if for all $y \in L$, we have
$y = \bigsup \{x \in L \mid x \prec y\}$.  The following sup-inf characterizations of continuity, hypercontinuity and prime continuity give an immediate insight into the relations among these different notions of continuity:
\begin{thm} \label{thm: 3 characterisations}
Let $L$ be a complete lattice.
\begin{enumerate}
\item $L$ is continuous if and only if for all $x \in L$,
\[
x = \bigsup \{\biginf U \mid x \in U \in \sigma(L)\}.
\]
\item $L$ is hypercontinuous if and only if for all $x \in L$,
\[
x = \bigsup \{\biginf (L \backslash \da M) \mid M \text{ is a finite subset of } L, x \not \in \da M\}.
\]
\item $L$ is prime continuous if and only if for all $x \in L$,
\[
x = \bigsup \{\biginf (L \backslash \da y) \mid x \not \in \da y\}.
\]
\end{enumerate}
\end{thm}
\begin{proof}
(1) follows directly from the basic definitions of Scott-open set and way-below relation.
The proofs for (2) and (3) can be found at \cite[p.509, Proposition VII-3.3]{gierzetal03} and
\cite{raney52}, respectively.
\end{proof}
Hence we have the following chain:
\begin{equation} \label{eq: prime implies hyper implies cts}
\text{prime continuity} \implies \text{hypercontinuity} \implies \text{continuity}.
\end{equation}

Hypercontinuity and quasicontinuity are connected via the following crucial result:
\begin{thm} (\cite{maoxu09})
\label{thm: spectrum of hypercontinuous}
The following are equivalent for a poset $P$:
\begin{enumerate}
\item $P$ is a quasicontinuous poset.
\item $\sigma(P)$ is a hypercontinuous lattice.
\end{enumerate}
\end{thm}

\section{Main results}
\label{sec: main results}

\begin{lem} \label{thm: finite reduce to one} Let $L$ be a
  join-continuous complete lattice. Then for any finite set $M =
  \{m_1,\ldots,m_n\} \subseteq L$, the following equation holds:
  \[
  \biginf (L \backslash \da M) = \bigsup_{k=1}^n \biginf (L \backslash \da m_k).
  \]
\end{lem}
\begin{proof}
  Note that the statement holds in the case that $M$ is empty as both
  sides then reduce to the least element of~$L$. If $M$ is nonempty
  then we use induction on~$n$. For $n = 1$ the equation is trivially
  true, so assume that the equation holds for all nonempty finite sets
  with $n$ elements; we must show that
  \[
  \biginf (L \backslash \da \{m_1,\ldots,m_n,m_{n+1}\}) = \bigsup_{k=1}^{n+1} (\biginf L \backslash \da m_k).
  \]
  By the induction hypothesis,
  \[
  \bigsup_{k=1}^n \biginf (L \backslash \da m_k) = \biginf (L \backslash \da \{m_1,\ldots,m_n\}).
  \]
  Thus, we have:
  \[
  \begin{split}
    \bigsup_{k=1}^{n+1} \biginf (L \backslash \da m_k)
    & = \left(\biginf (L \backslash \da m_{n+1})\right) \vee \biginf (L \backslash \da \{m_1,\ldots,m_n\}) \\
    & = \biginf \{\left(\biginf (L \backslash \da m_{n+1})\right) \vee s
    \mid s \in \bigcap_{i=1}^n (L \backslash \da m_i)\} \\
    & \quad
    (\text{Note:~} L \backslash \da \{m_1,\ldots,m_n\} = \bigcap_{i=1}^n (L \backslash \da m_i).)\\
    & = \biginf \left\{\biginf \{r \vee s \mid r \in (L \backslash \da m_{n+1}) \} \mid s \in \bigcap_{i=1}^n (L \backslash \da m_i) \right\} \\
    & = \biginf \left\{r \vee s \mid r \in (L \backslash \da m_{n+1}) \text{ and } s \in \bigcap_{i=1}^n (L \backslash \da m_i) \right\} \\
  \end{split}
  \]
  where join continuity is applied twice to obtain the second and
  third equalities. We finish the proof by showing that the set~$X$
  over which the infimum is taken in the last term is the same as $Y=L
  \backslash \da \{m_1,\ldots,m_{n+1}\}$. Indeed, an element of $X$ is by
  construction not below any of the elements $m_1,\ldots,m_{n+1}$, so
  we have $X\subseteq Y$. On the other hand, for any element~$t \in Y$, $t$ is
  in both $L\backslash\da m_{n+1}$ and $\bigcap_{i=1}^n (L \backslash \da
  m_i)$, so $t = t \vee t$ also belongs to $X$.  Thus, $X = Y$.

\end{proof}

\begin{thm} \label{thm: join cont + hyper cont = cdl}
The following statements are equivalent for a lattice $L$:
\begin{enumerate}
\item $L$ is join continuous and hypercontinuous.
\item $L$ is prime continuous.
\end{enumerate}
\end{thm}
\begin{proof}
  By Theorem~\ref{eq: prime implies hyper implies cts} and
  Theorem~\ref{eq: prime implies join}, we have (2) $\implies$ (1).
  So, it remains to show that (1) $\implies$ (2). To this end, by
  virtue of Theorem~\ref{thm: 3 characterisations}(3), we only need to
  show that for any $x \in L$, we have $x = \bigsup \{\biginf (L
  \backslash \da y) \mid x \not \in {\da y}\}$. Since $L$ is
  hypercontinuous, we have $x = \bigsup \{\biginf (L \backslash \da M)
  \mid M \text{ is a finite subset of } P, x \not \in \da M\}$. But
  for each finite set $M$ with $x \not \in \da M$, by Lemma~\ref{thm:
    finite reduce to one} we can write $\biginf (L \backslash \da M)$ as
  the supremum of terms of the form $\biginf (L \backslash \da m)$ with $m \in M$. Hence
  \[
  x = \bigsup \left\{\biginf (L \backslash \da M) \mid M \text{ is a finite subset of } P, x \not \in \da M\right\} =
  \bigsup \left\{\biginf (L \backslash \da y) \mid x \not \in \da y\right\}.
  \]
\end{proof}

\begin{thm}
\label{thm: main result}
The following statements are equivalent for a poset $P$:
\begin{enumerate}
\item $P$ is meet continuous and quasicontinuous.
\item $P$ is continuous.
\end{enumerate}
\end{thm}
\begin{proof}
  By Theorem~\ref{thm: scott open of meet cont is join cont}, $P$ is
  meet continuous if and only if $\sigma(P)$ is join continuous, and
  by Theorem~\ref{thm: spectrum of hypercontinuous}, $P$ is
  quasicontinuous if and only if $\sigma(P)$ is hypercontinuous. Thus,
  by Theorem~\ref{thm: join cont + hyper cont = cdl}, (1) is
  equivalent to $\sigma(P)$ being prime continuous, which, by
  Theorem~\ref{thm: spectra of completely distributive}, is equivalent
  to $P$ being continuous.
\end{proof}

\begin{remark}
  In Kou's original proof of the dcpo version of the above theorem,
  \cite{kouliuluo03}, it was argued that the lattice of Scott-closed
  sets $\Gamma(P)$ is continuous if $P$ is meet continuous. But this
  is not true in general. Construct a non-continuous frame of opens $P
  := \mathcal{O}(X)$ for your favourite non-locally compact space (see
  \cite[p. 417, Theorem V-5.5]{gierzetal03}). Now, were it the case
  that $\Gamma(P)$ is continuous for this choice of~$P$, then by
  Theorem~3.11 of \cite{hozhao09} it would follow that $\Gamma(P)$ is
  prime continuous. This would imply, by Theorem~\ref{thm: spectra of
    completely distributive}, that $P$ is continuous, a contradiction.
\end{remark}

\end{document}